\documentclass[12pt,fleqn,reqno]{amsart}

\usepackage{epsfig}
\usepackage{graphicx}
\usepackage{enumerate}
\usepackage{color} 
\usepackage[toc,page]{appendix}

\usepackage[normalem]{ulem}
\usepackage{cancel}

\usepackage{mathabx}

\usepackage[dvipsnames]{xcolor}
\usepackage{bm,bbm}
\usepackage{cancel}
\usepackage{enumitem}

\usepackage{physics}

\usepackage{amsrefs} %
\usepackage{euscript}

\usepackage{graphicx}

\usepackage{xcolor}

\usepackage{enumerate}
\usepackage{amsfonts,amssymb,amsthm,amsmath} 
\usepackage{hyperref}

\usepackage{enumitem} 
\usepackage{xspace}
\usepackage[margin=1.37in]{geometry}

\usepackage{todonotes} 
\usepackage{float} 

\usepackage{soul}

\newcommand{\stkout}[1]{\ifmmode\text{\sout{\ensuremath{#1}}}\else\sout{#1}\fi}

\newtheorem{thm}{Theorem}

\newtheorem{lemma}[thm]{Lemma}
\newtheorem{prop}[thm]{Proposition}
\newtheorem{cor}[thm]{Corollary}

\theoremstyle{remark}
\newtheorem{remark}[thm]{Remark}

\theoremstyle{definition}

\numberwithin{thm}{section}
\numberwithin{equation}{section}

\usepackage{soul}

\definecolor{green}{rgb}{0.0, 0.5, 0.5}
\definecolor{yellow}{rgb}{0.5, 0.5, 0}
\definecolor{lgray}{gray}{0.9}
\definecolor{llgray}{gray}{0.95}
\definecolor{lllgray}{gray}{0.975}
\newcommand{\red}{\color{red}}

\newcommand{\nc}{\newcommand}

\newcommand{\R}{\mathbb{R}}

\nc{\al}{\alpha}
\nc{\be}{\beta}
\nc{\G}{\Gamma}
\nc{\et}{\eta} 
\nc{\g}{\gamma}
\nc{\gam}{\gamma}
\nc{\ka}{\kappa}
\nc{\lam}{\lambda}
\nc{\Lam}{\Lambda}
\nc{\Om}{\Omega}
\nc{\om}{\omega}

\nc{\ta}{\tau}
\nc{\w}{\omega}
\nc{\io}{\iota}
\nc{\z}{\zeta}
\nc{\s}{\sigma}
\nc{\Si}{\Sigma}
\nc{\vphi}{\varphi}
\nc{\e}{\epsilon}

\nc{\bP}{\bar{P}}
\nc{\bQ}{\bar{Q}}

\nc{\ran}{\rangle}
\nc{\lan}{\langle}

\newcommand{\ra}{\rightarrow}

\newcommand{\ls}{\lesssim}

\newcommand{\re}{\operatorname{Re}}
\newcommand{\im}{{\rm Im}}

\newcommand{\dist}{\mathrm{dist}}

\nc{\bfone}{{\bf 1}}
\nc{\1}{{\bf 1}}

\newcommand{\p}{\partial}

\newcommand{\n}{\nabla}

\newcommand{\DETAILS}[1]{}

\newcommand{\x}{\lan x\ran}
\newcommand{\dx}{d_X (x)}

\makeatletter
\let\process@citelist\process@citelist@unsorted
\makeatother

\begin{document}

\rightline {\small \emph{Published in}}  
\rightline{\small \emph{Lett Math Physics, 2025}}
\bigskip

\title[QI propagation]{On Propagation of Information in Quantum Mechanics and Maximal Velocity Bounds}
\author[I.~M.~Sigal]{Israel Michael Sigal}
	\address{Department of Mathematics, University of Toronto, Toronto, ON M5S 2E4, Canada }
	\email{im.sigal@utoronto.ca}
	
	\author[X.~Wu]{Xiaoxu Wu}
	\address{Mathematical Sciences Institute, Australian National University, Acton ACT 2601, Australia}
		\email{Xiaoxu.Wu@anu.edu.au}

	\date{\today}
	\subjclass[2020]{35Q40   (primary); 81P45   (secondary)}
	\keywords{Maximal propagation speed; Lieb-Robinson bounds; quantum mechanics; quantum information; quantum light cones,  quantum observables quantum states}
	
	\maketitle

	\begin{abstract}
	We revisit key notions related to the evolution of quantum information in few-body quantum mechanics {(fbQM)} and, for a wide class of dispersion relations, prove uniform bounds on the maximal speed of propagation of quantum information for states and observables with exponential error bounds. Our results imply, in particular, a {fbQM} version of the Lieb-Robinson bound, which is known to have wide applications in quantum information sciences. We propose a novel approach to proving maximal speed bounds.

	\end{abstract}

	\section{Introduction}
	\subsection{Problem and results} The study of evolution of information in condensed matter physics is an active, robust area of research with many profound results. At the same time, perhaps due to the difficulty of experimental implementation, with the exception of a few works on quantum open systems (\cites{Breteaux_2022, Breteaux_2023}), this fundamental issue was not tackled in the original setting of quantum mechanics, i.e. at zero particle density. In this paper, we address this subject in a systematic way.

Investigation of propagation of quantum information has begun in the context of condensed matter physics with the discovery ( {  \cites{H1, H0,BHV, EisOsb,H3,NachOgS, NachS,HastKoma,H2}}) that the Lieb-Robinson bound obtained for lattice spin systems in statistical mechanics can be used to derive general constraints on propagation of quantum information. Time bounds on quantum messaging, creation and propagation of corrections and entanglement, state transport and control, quantum simulation algorithms, belief propagation raised in these papers were improved and extended significantly in  \cites{BD, BH, CL, DefenuEtAl, EisOsb, EldredgeEtAl, EpWh, FaupSig, KGE, KS1, NachSchlSSZ, NSY2,	NachVerZ, ElMaNayakYao,FLS1,FLS2,Fossetal,GebNachReSims, LRSZ,GEi, KuwLem2024,KuwSaito,KuwSaito1,TEtal4, KVS, MatKoNaka,NRSS,NachSim,SHOE,SZ,WH,YL,  Pou,RS, TranEtAl3,	TranEtAl1,TranEtal5}, see the survey papers \cites{Bose,	GEi, NachSim} and brief reviews in \cites{FLS1, FLS2, KuwSaito1}. \par

A different approach was introduced in \cite{SigSof2} and extended in \cites{AFPS,APSS,BoFauSig, FrGrSchl, HeSk, HunSigSof, SchSurvey, Sig, Skib}. Dealing originally with scattering theory in quantum mechanics, it was extended to many-body systems proving light-cone bounds on the propagation in bose gases (\cites{FLS1, FLS2,LRSZ, SZ}), the problem which was open since the groundbreaking work of Lieb and Robinson (\cite{LR}).\par



In this paper, we consider the evolution of states and observables in the quantum-mechanical context and prove the uniform maximal velocity bound with exponential tails yielding existence of the effective light cone (LC) (modulo exponentially small leakage) in quantum mechanical systems. In particular, we prove the Lieb-Robinson bound implying 
 the simultaneous measurability of evolving observables as long as their light cones do not intersect. 


\subsection{Setup}Consider a quantum system with a state space $\mathcal H$ and a Hamiltonian $H$, a self-adjoint operator on $\mathcal H$. We suppose that $\mathcal{H}=L^2(\Lam)$, where $\Lam$ is either $\mathbb{R}^n$ or $\mathbb Z^n$ or a bounded subset (box) in $\mathbb{R}^n$ or $\mathbb Z^n$. To fix ideas, in what follows, we take $\Lambda=\mathbb R^n$. 
\par


A specific operator $H$ we have in mind is the Schr\"odinger-type operator  
 \begin{align} \label{H}H=\om(p)+V(x),\end{align} 
 where $\om(k)$ is a real smooth positive function, $p:=-i\n$ is the momentum operator and the potential $V(x)$ is a real function s.t. $H$ is self-adjoint on the domain of $\om(p)$, i.e. $V(x)$ is $\om(p)$-bounded with the relative bound $<1$. We will also require certain analytic properties for the dispersion law $\omega(k)$ which can be interpreted as $\omega(k)$ having effectively bounded group velocity.

 \par
 
 
One can recognize information by the properties that it is transmittable, deletable\footnote{One should be able to erase parts of information one processes (i.e. irrelevant or inaccessible parts)}, localizable and measurable (or at least detectable).


For quantum mechaniccal systems, the second property requires extending the state space $\mathcal H$ to the space, $S_1^+$, of positive, trace-class operators, $\rho$, acting on $\mathcal H$. The original state space $\mathcal H$ is identified with the subspace of rank one projections. Quantum information related to a given system is encoded in density operators describing it. \par


The evolution of density operators, is given by the von Neumann equation (vNE) (here and in the rest of this section we set $\hbar=1$) 
 \begin{align}\label{vNeq}
	&\frac{\partial\rho_t}{\partial t}=-i[H,\rho_t],\ \quad \text{ with }\ \rho_{t=0}=\rho_0.\end{align}
This equation preserves the rank of projections and, when reduced to the rank-one projections, is equivalent to the Schr\"odinger equation.\par


The initial value problem~\eqref{vNeq} has a unique solution which generates the automorphism on $S_1^+$:
\begin{equation}
    \alpha'_t(\rho):=e^{-iHt} \rho e^{iHt}.\label{IVP}
\end{equation}
This formula allows one to reduce many results on propagation of states and observables to estimates on the Schr\"odinger evolution $e^{-iHt}$.


However, the framework of the vNE is much broader and entails a different take on the evolution problem (i.e. the semiclassics for the vNE leads to Liouville's equation, rather than Newton's one), and it is foundational for the theory of open quantum dynamics.


\subsection{MVB}We begin with a key result concerning the propagator $e^{-itH}$, which implies a variety of bounds on propagation of quantum information.\par

In what follows, we always assume (without specifying this) that the quantum Hamiltonian $H$ we deal with is self-adjoint.  For self-adjoint operators on $L^2(\mathbb R^n)$, we use the notation $\sup A=\sup\limits_{u\in D(A), \|u\|=1} \langle u, Au\rangle$. For families of bounded operators, we use the notion of analyticity in the sense of Kato (resolvent sense, see \cite{RS4}, Section XII). \par


Let $T_\xi$ be the unitary operator of multiplication by the function $e^{-i\xi\cdot x}$. For a self-adjoint operator $H$ on $L^2(\R^n)$, we introduce the operator family $H_\xi:=T_\xi H T_\xi^{-1}$. 
Let $ S_a^n=\{\zeta=(\zeta_1,\cdots,\zeta_n)\in \mathbb C^n\, :\, \, |\im z_j|<a \quad \forall\, j\}. $ Now, we assume 
\begin{itemize}[label={},ref={}]
\item[(A)]\label{asp: An} The family $H_\xi:=T_\xi HT_\xi^{-1}$, $\xi\in \mathbb R^n$, has an analytic continuation in $\xi$ from $\mathbb R^n$ to $\mathcal S_a^n$ in the sense of Kato and for this continuation $\im H_\zeta=\frac{1}{2i}(H_\zeta-H_\zeta^*),\, \zeta\in \mathcal S_a^n$, are bounded operators.
\end{itemize}
\bigskip


We fix $\mu\in (0,a)$ and define the number  
\begin{equation}\label{c}
  c:=\sup\limits_{\xi\in \mathbb R^n,\, b\in \mathcal S^{n-1}} \sup(\im H_{\xi+i\mu b})/{\mu}\    { =\sup\limits_{y\in \mathcal \R^{n}, |y|=\mu} \sup(\im {H}_{iy})/{\mu}}.
\end{equation}


In what follows, $X$ and $Y$ denote open subsets of $\R^n$, $X^c:=\R^n- X$, $d_{XY}$, the distance between $X$ and $Y$ and $\chi_{X}$, the characteristic function of  $X$, as well as  the operator of multiplication by this function. Moreover, depending on the context, $\|\cdot\|$ stands either for the norm in $L^2(\mathbb R^n)$, or the operator norm on $L^2(\mathbb R^n)$. We have the following result:  
  
\begin{thm}[Light cone (maximal propagation velocity) bound] \label{thm:MVB}
Let Condition~(A) hold and let $\mu\in (0,a)$. Then,  for any $\mu'\in (0,\mu)$ and for any two disjoint sets $X$ and $Y$ in $\R^n$, we have, 
\begin{equation} \label{MVB} 
 \|\chi_{_{X}}\,e^{-iH t} \chi_{_{Y}}\|\leq \, C  {e^{-\mu' (d_{XY}- c'|t|)}}, 
 \end{equation} 
 where $c'=\frac{\mu c}{\mu'}$, with $c$ given in Eq.~\eqref{c}, 
 and $C>0$ is a constant depending on $\frac{\mu}{\mu'}-1,\ \mu$ 
 and $n$. 
    \end{thm} 


This theorem is proven in Section~\ref{sec:pfMVB}. We call inequality~\eqref{MVB} the \emph{uniform maximal velocity bound} (uMVB).\par 

Inequality~\eqref{MVB} can be interpreted as the Hamiltonian having effectively finite group velocity (defined as $i[H, x]$). For unbounded group velocities, MVB is not true unless one restricts initial conditions.\par
For the Schr\"odinger-type operator \eqref{H}  on $L^2(\R^n)$, Assumption~(A) follows from the following condition



\begin{align}\label{om-cond'} 
\begin{aligned}
    &\hspace{-1cm}\text{$\omega(k)$ has an analytic continuation ($\omega(\zeta)$) from $\mathbb{R}^n$ to \st{the polystrip} $\mathcal{S}_a^n$,} \\
    &\hspace{-1cm}\text{\st{for some $a>0$,} and $\im \omega(\zeta)$ is a bounded function on $\mathcal{S}_a^n$.}
\end{aligned}
\end{align}
\medskip


Note that the Schr\"odinger operator $H=-\Delta+V$  on $L^2(\mathbb R^n)$ does not satisfy the second part of Condition~\eqref{om-cond'} since $\im (p+\zeta)^2=2 p\cdot \im \zeta $ is not bounded.  On the other hand, the semi-relativistic Hamiltonian 
\[H=\sqrt{-\Delta+m^{2}}+V\] 
obeys Condition \eqref{om-cond'} with $a=m$.\par 
Earlier, state-dependent, power-decay MVB were proven, in connection with quantum scattering theory,  in \cites{APSS, HeSk, HunSig2, SigSof2,Skib}. As we were completing this paper, we came across a very recent preprint \cite{CJWW} which proves MVB with exponential estimates on tails for (translationally invariant) quantum dynamics of a particle on a lattice and whose approach overlaps with ours.\par
As was mentioned above, Theorem~\ref{thm:MVB} holds with $\mathbb R^n$ replaced by $\mathbb Z^n$ and covers   Hamiltonians on $L^2(\mathbb Z^n)$ of the form
\begin{equation}
    H=T+V,
\end{equation}
 where $T$ is 
  a {symmetric } operator with an exponentially decaying matrix elements $t_{ x, y}$, i.e.
\begin{equation}
    |t_{ { x, y}}|\leq Ce^{-a|\ { x - y}|}, \qquad \text{for some }a>0,
\end{equation}
e.g.  the negative of the discrete Laplacian $\Delta_{\mathbb Z^n}$ on $\mathbb Z^n$. 
In this case, $\xi$ in the definition of $T_\xi$ and in Condition (A) {could be identified with a point in} the dual (quasimomentum) space $K\equiv (\mathbb Z^n)^*\approx\ 
 \mathbb R^n/ \mathbb Z^n \stkout{\approx F}$, and $\xi \cdot x$ {could} be understood as the linear functional $\xi(x)\in K$. (For a general lattice $\mathcal L$ in $\mathbb R^n$, the (quasi) momentum space $\mathcal L^*$ is isomorphic to the torus $\mathbb R^n/ \mathcal L'$, where $\mathcal L'$ is the lattice reciprocal to $\mathcal L$.)


{Theorem \ref{thm:MVB-diff} of Appendix \ref{sec: DEs} implies} that integral kernels decaying as powers lead to the power decay of corrections.

One can also extend our proof to matrix-valued Hamiltonians accounting for {internal} degrees of freedom {and to time-dependent Hamiltonians, $H(t)$.}


Finally, we mention 
the minimal velocity bounds used extensively in the scattering theory (see \cites{DHS, Der,  DerGer, FS0, FaupSiga, FLS1, FrGrSchl, FrGrSchl2, Ger, Gries, HeSk, HunSigSof, Sig, SigSof2, Skib} and reviews in \cites{DerGer2, HunSig2}). The latter involves a power decay of the leakage  
and is proven on the infinitesimal level, via lower bounds on commutators of $H$ with appropriate 
 generators, say, $A=\frac12(\n \om(p)\cdot x+ x\cdot \n\om(p))$, 
restricted to thin energy shells away from the zero group velocity regions (the Mourre estimate). This leads to severe restrictions on the potentials, $V(x)$, and dispersion relations, $\om(p)$. 


\DETAILS{
On a level of deformations, 
this would involve  proving an estimate of the form $\|e^{-i H_\z t}\|\le C e^{-ct}$, for some $C, c>0$,  on the deformed semigroup $e^{-i H_\z t}$, where $H_\z$ is an analytic continuation of the family $H_\xi:= T_\xi H T_\xi^{-1}$, where $T_\xi:=e^{i A \xi}$, as per our proof. 
\vspace{3cm}}



\subsection{Localization}A key notion in analysis of evolution of quantum information is that of localization. It is reasonable to consider states localized in bounded sets, say, states created in an apparatus in a lab. With this motivation, we say that 
\begin{itemize}
\item a state $\rho$ is \emph{localized} in $X$, if in $\rho$, the probability for the system to be in $X$ is equal to $1$:
\begin{align}\label{st-loc} \Tr( \chi_{X}\rho)=1 \quad \text{ or }\ \quad  \Tr( \chi_{X^c}\rho)=0.
\end{align}
\end{itemize}
\begin{remark} By linearity, this notion could be readily extended to the one of \emph{locally perturbed states}. 
\end{remark}


\subsection{Light cone for evolution of states and observables}

\begin{cor}\label{cor: 1.2} Suppose Condition~(A) holds. Then, for any density operator $\rho_0$ localized in $X$, the probability that its evolution $\rho_t=\alpha_t'(\rho_0)$ is localized in a disjoint set $Y$ is bounded as 
\begin{equation}
    \Tr(\chi_Y \alpha'_t(\rho_0))\leq Ce^{-\mu'(d_{XY}-c't)},
\end{equation}
where $\mu, c,\mu',c'$ are as in Theorem~\ref{thm:MVB} and $C=C(\frac{\mu}{\mu'}-1, \mu, n)>0$.
\end{cor}


This corollary says that the probability that $\rho_t$ spills outside the light cone 
\begin{equation}
\Lam_{X,c}:=\{(x, t): \dx< ct\}\label{LamX}
\end{equation}
of $X$ is exponentially small.\par


The second key ingredient in the general theory is the notion of observables. These are operators on $\mathcal H$ representing actual physical quantities and their probability distributions. An average of a physical quantity (say, momentum) represented by an observable $A$ (say, $p=-i\nabla$) in a state $\rho$ is given by $\Tr(A\rho)$. There is a duality between states and observables given by the coupling 
\begin{equation}
   (A,\rho)\equiv \rho(A):= \Tr(A\rho),
\end{equation}
which can be considered as either a linear, positive functional of $A$ or a convex one of $\rho$. In what follows, we identify density operators $\rho$ with linear positive functionals $\rho: A\ra \rho (A):=\Tr( A\rho)$.\par


The evolution of observables is determined by the Heisenberg equation
\begin{equation}
    \p_t A_t=i [H, A_t].\label{eq: ptAt}
\end{equation}
Given initial conditions $A_t\vert_{t=0}=A$, ~\eqref{eq: ptAt} generates the (Heisenberg) automorphism 
\begin{equation}
   A_t\equiv \alpha_t(A)=e^{itH} Ae^{-itH}.
\end{equation}
The Heisenberg equation (or representation) is equivalent to the vNE and, since observables form $C^*$ algebra, often, is more convenient to work with. The duality between states and observables extends to respective evolutions 
\begin{equation}
    \Tr(\alpha_t(A)\rho)=\Tr(A\alpha_t'(\rho)).
\end{equation}
\par It is natural to have observables which act locally, i.e. in some set, but leave states outside this set unchanged. Thus, we introduce
\begin{itemize}
    \item an observable $A$ \emph{acts} on $X$ iff it is of the form 
\begin{equation}
    A=\chi_XA\chi_X+\chi_{X^c},\label{A=Ax+Axc}
\end{equation}
where, recall, $\chi_{X}$ and $\chi_{X^c}$ stand for multiplication operators by the corresponding cut-off functions (so that $\chi_X+\chi_{X^c}=\mathbbm1$). \par 
\end{itemize}\par
As suggested by the term,~\eqref{A=Ax+Axc} implies that $\chi_{X^c} A\psi=\chi_{X^c}\psi$ and $A\psi=A\chi_X\psi+\chi_{X^c}\psi$. Note that if $A$ and $B$ act on $X$ and $Y$, respectively, then 
\begin{equation}
[A,B]=0, \quad \text{ whenever $X\cap Y=\emptyset$.}
\end{equation}\par


We call the smallest set on which an observable $A$ acts the \emph{action domain} of this observable and denote it by $\text{act} A$. Many notions related to and statements about evolution of quantum information can be formulated in parallel for states and observables.


The next useful result on localization of the Heisenberg evolution parallels Corollary~\ref{cor: 1.2}. It shows that the evolution $A_t=\alpha_t(A)$ acts, up to exponentially small tails, within the light cone of its initial action domain of $A$.



We define an approximation of the evolution $A_t=\alpha_t(A)$ in the set $U$ as
\begin{equation}
    A_{t,U}:=\chi_U A_t\chi_U+\chi_{U^c}.\label{def: AtY}
\end{equation}
Clearly, $A_{t,U}$ act on the set $U$. For a subset $X\subset \mathbb R^n$, let $X_\eta$ be the $\eta$-neighborhood of $X:$  
\begin{equation}
X_\eta =\{x\in \mathbb R^n\,:\, d_X(x)< \eta\}.  
\end{equation}
\begin{thm}[Light cone approximation of Heisenberg evolution]\label{thm: LCA} Suppose the Condition~(A) holds. Then there exist $c>0$ and $C=C(c,n)>0$ s.t. for any $\eta\geq 1$ and every open set $X$ and every operator $A$ acting on $X$, the evolution $A_t\equiv \alpha_t(A)$ satisfies 
\begin{equation}
    \|A_t-A_{t,X_\eta}\|\leq Ce^{-\mu (\eta-ct)}\|A\|.\label{aprxAt}
\end{equation}
\end{thm}
This theorem is proven in Section~\ref{sec: MVB}. It says that, up to exponentially small tails, the evolution of an operator acting on $X$ acts inside the $c$-light cone of $X$. \par 
One can also define the localization of the evolution $A_t$ to a set $U$ as $A_{t,U}=e^{itH_{U}}Ae^{-itH_{U}}$ as an evolution of $A$ with a Hamiltonian $H_{U}$ supported in $U$, but the proof of this version is more involved. \par


Our next result has no parallel for quantum states and is an analogue of one of the key results of quantum information theory.


\begin{thm}[Quantum-mechanical Lieb-Robinson bound] 
	\label{thm:LRB}
Suppose Condition~(A) holds and let $X,Y\subset \mathbb R^n$ with $d_{XY}>0$. {Then, there exist $c>0$ and $C=C(n,c)>0$ s.t.~for }every pair operators $A$ and $B$ acting on $X$ and $Y$, respectively, we have the following estimate 
	\begin{align}\label{LRB'}\|[ \al_t (A), B]\| \le C e^{-\mu (d_{XY}-ct)}\|A\|\|B\|.
\end{align}	\end{thm} 
A proof of Theorem \ref{thm:LRB} is given in Section \ref{sec:LRB}. We call \eqref{LRB'} the \textit{quantum-mechanical Lieb-Robinson bound (LRBqm)}. Recall that the physical importance of commuting quantum observables is that they can be measured simultaneously.


Theorem~\ref{thm:LRB} yields that, for any state $\rho$ and for all $\abs{t}< d_{XY}/c$,
\begin{align}\label{LRB}		|\rho ([ \al_t (A), B] )| 
	\le C\|A\|\|B\| e^{-\mu (d_{XY}-ct)}. 
	\end{align}

\begin{figure}[h!]
    \centering
\begin{tikzpicture}
    \draw[thick] (3.5,0) -- (5,4);
    \draw[thick] (2,0) -- (0.5,4);
    \node at (2.8,2.5) {light cone};
    \draw[dashed] (0,2) -- (10.5,2);
    \node at (2.8,2) {of $A$ (or $X$)};

    \draw[thick] (6.5,0) -- (6.5,4);
    \draw[thick] (9.5,0) -- (9.5,4);
    
    \node at (8,2.5) {space-time};
    \node at (8,2) {dynamics of $B$ };
    \node at (7.7,1.5) {(or $Y$)};
    \draw[very thick] (4.25,2) -- (6.5,2) node[midway, below] {$d_{XY}-ct_0$};
    \draw[thick,->] (5.375,2) -- (5.375,1.75);

   \begin{scope}
        \clip (3.5,0) -- (5,4) -- (0.5,4) -- (2,0) -- cycle;
        \foreach \x in {0,0.1,...,10} {
            \draw[thin] (\x,0) -- (\x-3.5+0.5,4);
        }
    \end{scope}
    \begin{scope}
        \clip (6.5,0) -- (6.5,4) -- (9.5,4) -- (9.5,0) -- cycle;
        \foreach \x in {0,0.1,...,14} {
            \draw[thin] (\x,0) -- (\x-3.5+0.5,4);
        }
    \end{scope}

    \draw[very thick] (2,0) -- (3.5,0) node[midway, below] {$X$};
    \node at (0.3,4.8) {$t$};
    \draw[->] (0,0) -- (0,5);
    \draw[thin] (0,0) -- (2,0);
    \draw[thin] (3.5,0) -- (6.5,0);
    \draw[very thick] (6.5,0) -- (9.5,0) node[midway, below] {$Y$};
    \draw[thin] (9.5,0) -- (10.5,0) node[right] {$\mathbb R^n$};
    \node at (-0.3,2) {$t_0$};


\end{tikzpicture}
\caption{Light cone diagram of $A$ and $B$}
    \label{fig: 2-1}
\end{figure}
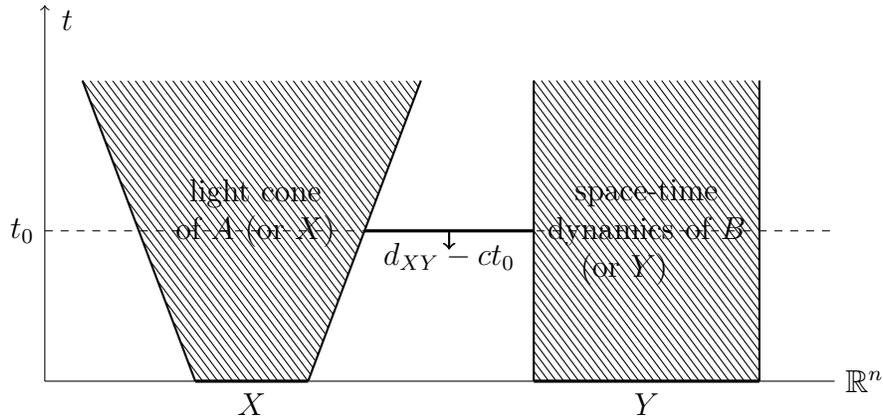
\noindent Estimate \eqref{LRB} shows that, with the probability approaching $1$ exponentially, as $t\ra \infty$, an evolving observable $A_t=\al_t(A)$ commutes with any other observable acting outside the light cone
	\[\{ \,(x,t) \,\,|\, \dist(x, \text{act} A)\le c t\},\]
 of $\text{act} A$, where $\text{act} A$ is the smallest set on which $A$ acts (see Fig.~\ref{fig: 2-1}).


	\DETAILS{This implies that the maximal speed of quantum propagation for \eqref{HeisEvol} is bounded  (up to an absolute constant) by the number $\kappa$ defined in \eqref{kappa}.}

Recall that the expectation  $-\rho([A(t), B]^2)\equiv -\tr([A(t), B]^2 \rho)$ is called the  out-of-time-order correlations (OTOC). The inequality $-\rho([A(t), B]^2)\le \|[A(t), B]\|^2$ and Theorem~\ref{thm:LRB} imply   
\begin{cor}[OTOC estimate] 
	\label{cor:OTOCest} Suppose the Condition~(A) holds and let $\mu$ and $c$ be as in Theorem~\ref{thm:LRB}. Then OTOC $-\rho([A(t), B]^2)$ satisfies the estimate 
\begin{align}\label{OTOCest}-\rho([A(t), B]^2)\le  C e^{-2\mu (d_{XY}-ct)}(\|A\|\|B\|)^2. 
\end{align}
\end{cor}
Estimate \eqref{OTOCest} extends an estimate of OTOC in the finite dimensions,
\begin{equation}\frac1{D}\tr([A(t), B]^2),\end{equation}
where $D=\dim\mathcal H$, see e.g. \cite{KuwSaito1}, to the infinite-dimensional case. \par


In Appendix~\ref{sec: DEs}, we extend our approach to differentiable deformations $H_\xi$ obtaining power bounds on the error terms. We expect that it could be extended to open quantum systems, where estimates of evolving states and observables cannot be reduced to estimate of the Schr\"odinger evolution and one has to estimate instead the von Neumann-Lindblad semigroup, $e^{Lt},$ see \cite{Breteaux_2022}, and to many-body (condensed matter) systems as suggested by our results on $N$-particle Schr\"odinger dynamcis presented in Section~\ref{sec: MVB-N}. \par


This paper is organized as follows. Theorems~\ref{thm:MVB},~\ref{thm: LCA} and~\ref{thm:LRB} are proven in Sections~\ref{sec:pfMVB},~\ref{sec: MVB} and~\ref{sec:localization}, respectively. In Appendices~\ref{sec: DEs} and~\ref{sec: MVB-N}, we present extensions of our technologies to differentiable families $H_\xi$ and to $N$-particle systems.\par
\underline{Notation.} We use the abbreviation $\| \cdot\|$ for both $ \| \cdot\|_{L^2(\mathbb{R}^n)}$ and $\| \cdot\|_{L^2(\mathbb{R}^n)\to L^2(\mathbb{R}^n)}$. Throughout this paper, $C$ will denote a constant and may vary from one line to another. We write $\lesssim$ or $\gtrsim$ whenever $A\leq CB$ or $CA\geq B$ for some constant $C>0$. We write $A\lesssim_a B$ or $A\gtrsim_a B$ if  $A\leq C_aB$ or $C_aA\geq B$ for some constant $C_a>0$ which depends on parameter $a$.

As usual, $\p_x^\alpha=\prod\limits_{j=1}^n \p_{x_j}^{\alpha_j}$, for $\alpha=(\alpha_1,\cdots,\alpha_n)$, $x=(x_1,\cdots,x_n)$ and $|\alpha|=\sum\limits_{j=1}^n\alpha_j$. \par


In what follows, $C$ stands for a generic constant which changes from equation to equation and are independent of variable parameters, such as distances between sets and their sizes, etc. \par


\section{Proof of Theorem~\ref{thm:MVB}}\label{sec:pfMVB}
\begin{proof}[Proof of Theorem~\ref{thm:MVB}] As usual for $\mathbb R^n$, we use the Euclidean inner product to identify the position and momentum spaces. \par



Recall the $n$-parameter unitary group 
\begin{equation}\label{def: Txi}
T_{\xi}=e^{-ix\cdot \xi}, \xi\in \mathbb{R}^n, 
\end{equation}
and introduce the {deformed evolution} 
\begin{equation}\label{def: Utxi: n}
U_{t,\xi}:=T_{\xi} e^{-itH}T_{\xi}^{-1}.
\end{equation}
By the unitarity of $T_\xi$, we have 
\begin{equation}
    U_{t,\xi}=e^{-i H_\xi t},\qquad\text{ where }H_\xi:=T_\xi HT_\xi^{-1}.
\end{equation}


\begin{prop}\label{prop: 2.6}
Under Assumption~(A), the operators $U_{t,\xi}$ have  analytic continuations in $\xi$ from $\mathbb R^n$ to $\mathcal S_a^n$, as bounded operators.  
\end{prop}
\begin{proof} By Assumption~(A), $H_\xi$ has an analytic continuation in $\xi$ from $\mathbb R^n$ to $\mathcal S_a^n$, in the sense of Kato, and the operator $\im H_\zeta=\frac{1}{2i}(H_\zeta-H_\zeta^*)$ is bounded for every $\zeta\in \mathcal S_a^n$. Hence, the relations $H_\zeta=H_r+iH_i,$ where $H_r=\re H_{\zeta}$ and $H_i=\im H_{\zeta}$, and
\begin{equation}
(H_{\zeta}-i\lambda )^{-1}=(H_r+iH_i-i\lambda )^{-1}=R_i^{1/2}\left[ R_i^{1/2}H_rR_i^{1/2}-i\right]^{-1}R_i^{1/2},\label{re: Hlambda}
\end{equation}
{where $R_i=(\lambda-H_i)^{-1}$, imply that $\sigma(-iH_{\zeta})\subset \{\, \zeta\in \mathbb C^n\, : \, |\re \zeta|\leq C\,\}$, $\forall\, \zeta\in \mathcal S_a^n$, for $C=\|\im H_{\zeta}\|>0$, and, for any $\lambda \in \mathbb C$ with $|\re\lambda|>C$, we have the estimate}
\begin{equation}
 \|(H_{\zeta}-i\lambda)^{-1}\|\leq (|\re\lambda|-\|\im H_{\xi^z}\|)^{-1}.
\end{equation}
Hence, by the Hille-Yosida theorem, $H_{\zeta}$ generates the bounded evolution $U_{t,\zeta}=e^{-iH_{\zeta}t}$, $t\in \mathbb R$, which is analytic as an operator-function of $\zeta\in \mathcal S_a^n$, by {the computation} 
{$\p_{\bar\z_j} U_{t,\z}=-i \int_0^t e^{-i(t-s)H_{\z}} \p_{\bar\z_j}H_{\z}e^{-i H_\z s}ds=0 \forall j .$}\end{proof}


 Let $U_t\equiv e^{-itH}$ and $U_{t,\xi}:=T_\xi U_t T_\xi^{-1}$. We have 
\begin{lemma}[Key lemma]\label{lem: 2.3}Under Condition~(A), we have
\begin{align}
    \chi_XU_t\chi_Y=& \chi_X T^{-1}_\zeta U_{t,\zeta} T_\zeta\chi_Y,\qquad \forall\, \zeta\in \mathcal S_a^n.\label{Ut: eq T}
\end{align}
\end{lemma}
\begin{proof} Using the invertibility of $T_\xi$ and the definition $U_{t,\xi}=T_\xi U_t T_\xi^{-1}$, we rewrite $\chi_X U_t\chi_Y$ as
\begin{align}
    \chi_XU_t\chi_Y=&\chi_X T_\xi^{-1} T_\xi U_t T_\xi^{-1}T_\xi \chi_Y= \chi_X T^{-1}_\xi U_{t,\xi} T_\xi\chi_Y,\qquad \forall\, \xi\in \mathbb R^n.\label{Ut: eq T1}
\end{align} 
The right-hand side of Eq.~\eqref{Ut: eq T1} has an analytic continuation in $\xi$ from $\mathbb R^n$ to $\mathcal S_a^n$ and is independent of $\xi\in \mathbb R^n$. Hence, its analytic continuation is independent of $\zeta\in \mathcal S_a^n$ and~\eqref{Ut: eq T} follows.
\end{proof}


Relation~\eqref{Ut: eq T} implies, for $\zeta\in \mathcal S_a^n$,
\begin{equation}
    \| \chi_XU_t \chi_Y\|\leq \| \chi_X T_\zeta^{-1}\|\| U_{t,\zeta}\|\|\chi_Y T_\zeta\|.\label{YUtX: est}
\end{equation}
Let $\zeta=\xi+i\mu b$, $\xi\in \mathbb R^n, \mu \in (0,a)$ and $b\in S^{n-1}$. Then estimate~\eqref{YUtX: est} implies 
\begin{equation}
     \|  \chi_X U_t\chi_Y\| \leq  e^{-\mu (r_X-r_Y)}\|U_{t,\zeta}\|,\label{eq: YUX}
\end{equation}
where
\begin{equation}
    r_Y:= \sup\limits_{y\in Y}b\cdot y\quad \text{ and }\quad r_X:= \inf\limits_{x\in X}b\cdot x.\label{eq: 2.10}
\end{equation}


We can cover $X$ and $Y$ by small balls. Hence, we begin with $Y=B_r(y_0)$ and $X=B_r(x_0)$ with $r=\frac{\epsilon}{2}d_{XY}$ for some $y_0\in Y$, $x_0\in X$ and $\epsilon\in (0,1)$.

Translate both balls by $y_0$ to place $y_0$ at the origin and $x_0$ at $x_0-y_0$. Let $S_{y_0}$ denote the corresponding shift operator. Then we have  
\begin{equation}
    S_{y_0}\chi_X U_t\chi_YS_{y_0}^{-1}=S_{y_0}\chi_XS_{y_0}^{-1}S_{y_0}U_tS_{y_0}^{-1}S_{y_0}\chi_YS_{y_0}^{-1}=\chi_{X^{y_0}}U_t^{y_0}\chi_{Y^{y_0}},\label{eq: 2.9}
\end{equation}
where $X^{y_0}=X-y_0=B_{r}(x_0-y_0)$, $Y^{y_0}=Y-y_0=B_r(0)$ and $U_t^{y_0}:=S_{y_0}U_tS_{y_0}^{-1}.$ Thus, it suffices to estimate the right-hand side of Eq.~\eqref{eq: 2.9}.

Now, we skip the superindex $y_0$, so we are back to Eq.~\eqref{Ut: eq T}, but with $X=B_r(x_0-y_0)$ and $Y=B_r(0)$.

Let $\zeta=\xi+i\mu b$, where $b=\frac{x_0-y_0}{|x_0-y_0|}$. (For each pair of points $x_0\in X$ and $y_0\in Y$, we choose a different analytic deformation $b$.) Since $|x_0-y_0|\geq d_{XY}$, by the definition of $r_X$ and $r_Y$, we have 
\begin{equation}
    r_X-r_Y\geq (1-\epsilon)|x_0-y_0|\geq (1-\epsilon)d_{XY}.\label{E10a}
\end{equation}

Eqs.~\eqref{eq: YUX} and~\eqref{E10a} yield 
\begin{align}
    \| \chi_X U_t\chi_Y\|\leq& 
   \|U_{t,\zeta}\| e^{-\mu(1-\epsilon) d_{x_0y_0}}.\label{1.32}
\end{align}

Note that the complex deformation $U_{t,\zeta}$ of the evolution operator changes from one pair of balls to another. The next proposition provides a uniform estimates of various such deformations. 
\begin{prop}\label{prop: E} Let Assumption~(A) be satisfied, and let $\mu\in (0,a)$ and $c$ be as in Eq.~\eqref{c}. Then we have the estimate
\begin{equation}
    \| U_{t,\xi+i\mu b}\|\leq e^{\mu tc},\qquad \forall\,\xi\in \mathbb R^n{, \quad b\in S^{n-1}}.\label{est: Utxi: zeta}
\end{equation}
\end{prop}


\begin{proof}[Proof of Proposition~\ref{prop: E}] Take $g\in L^2(\mathbb R^n)$. We denote $g_{t,\zeta}:=U_{t,\zeta}g, \zeta\in \mathcal S_a^n$, and compute 
\begin{align}
    \p_t\| g_{t,\zeta}\|^2=&-i\langle g_{t,\zeta}, (H_{\zeta}-H_{\zeta}^*)g_{t,\zeta} \rangle = 2\langle g_{t,\zeta}, \left(\im H_{\zeta}\right)g_{t,\zeta} \rangle.\label{E3: eq1}
\end{align}
\DETAILS{The second part of Condition~(A) implies that \st{for any $a'\in [0,a)$,} 
\begin{equation}\label{c'}
    \stkout{\sup\limits_{\zeta\in \mathcal S_{a'}^n} \sup\im H_\zeta <\infty.}\
     {\red c  = c (\mu)=\sup\limits_{\xi\in \mathbb R^n,\, b\in \mathcal S^{n-1}} \sup(\im H_{\xi+i\mu b})/{\mu}.}\end{equation}
{Indeed, since $\im H_{\zeta+\xi}=T_\xi \im H_\zeta T_\xi^{-1}\, \,\,$ for every $\xi\in \mathbb R^n$ and $\zeta\in \mathcal S_a^n$, we have that $\sup \im H_\zeta$ is independent of $\re \zeta$ and therefore}
\begin{equation}\label{ineq: 1.5}
{\sup\limits_{\stkout{\zeta\in \mathcal S_{a'}^n}\ {\red \xi\in \mathbb R^n,\, b\in \mathcal S^{n-1}}} \sup\im H_{\stkout{\zeta}\ {\red \xi+i\mu b}} =\sup\limits_{\stkout{\im \zeta \in [0,a']^n}\  {\red b\in \mathcal S^{n-1}}} \sup\im H_{\stkout{\zeta}\ {\red \xi+i\mu b}} <\infty. }
\end{equation}
{\red Hence,}}
For $\zeta=\xi+i\mu b$, by Eq.~\eqref{c}, 
we have, $\im H_\zeta\leq \mu c$.
This together with Eq.~\eqref{E3: eq1} implies
\begin{align}
    |\p_t \|g_{t,\zeta}\|^2|\leq & 2\mu c\|g_{t,\zeta}\|^2.\label{ptgtzeta: est}
\end{align}
Since $g_{t,\zeta}\vert_{t=0}=g$, this gives $ \|g_{t,\zeta}\|\leq \|g\| e^{\mu ct}$ yielding~\eqref{est: Utxi: zeta}. \end{proof}


Eqs.~\eqref{1.32} and~\eqref{est: Utxi: zeta} yield 
\begin{equation}
    \|\chi_X U_t\chi_Y\|\leq e^{\mu ct} e^{-\mu(1-\epsilon) d_{x_0y_0}}.\label{2.14a}
\end{equation}


Now, we return to arbitrary disjoint sets $X$ and $Y$. We cover $X$ and $Y$ by the balls, $ B_j^X=B_r(x_j), j=1,\cdots,N_1$ and $B_k^Y=B_k(y_k), k=1,\cdots,N_2$, in $\mathbb R^n$ of the radius $r=\frac{\epsilon}{2}d_{XY}$, centred at $x_j$ and $y_k$, respectively. $N_1$ and $N_2$ could be either finite or infinite. With this cover, we associate partitions of unity 
\begin{equation}
\chi_X=\sum\limits_{j=1}^{N_1} \chi_{j}^2\quad \text{ and }\chi_Y=\sum\limits_{k=1}^{N_2}\tilde\chi_{k}, \label{cpt3}
\end{equation}
where $\chi_j$ and $\tilde \chi_k$ satisfy~$\text{supp}(\chi_j)\subset B_r(x_j)$ and $\text{supp}(\tilde\chi_k)\subset B_r(y_k),$ $j=1,\cdots,N_1$ and $k=1,\cdots,N_2$. For each $g\in L^2(\mathbb R^n)$, we estimate, using Eqs.~\eqref{cpt3}, 
\begin{align}
    \| \chi_X e^{-itH}\chi_Yg\|^2= &\sum\limits_{k=1}^{N_1}\|\chi_{k}e^{-iHt}\chi_{Y}g\|^2 \leq &\sum\limits_{k=1}^{N_1}\left(\sum\limits_{j=1}^{N_2}\|\chi_{k}e^{-iHt}\tilde\chi_{j}g\|\right)^2,
\end{align}
where in the equality we used the first partition of unity and in the inequality, the second one.
By~\eqref{2.14a}, we have
\begin{align}
     \| \chi_X e^{-itH}\chi_Yg\|^2 \leq & C e^{2\mu ct} M(g),\label{last: eq}
\end{align}
where $M(g)$ is given by, with $\mu'=\mu(1-\epsilon)$, 
\begin{align}
    M(g):=&\sum\limits_{k=1}^{N_1} \left( \sum\limits_{j=1}^{N_2} e^{-\mu' d_{x_ky_j}} \|\tilde\chi_{j}g\|\right)^2\nonumber\\
    =&\sum\limits_{k=1}^{N_1}\sum\limits_{j=1}^{N_2}\sum\limits_{l=1}^{N_2} e^{-\mu'(d_{x_ky_j}+d_{x_ky_l})}\|\tilde\chi_{j}g\|\|\tilde\chi_{l}g\|.
\end{align}
To estimate $M(g)$, we use arithmetic mean inequality to obtain
\begin{align}
    M(g)\leq \sum\limits_{k=1}^{N_1}\sum\limits_{j=1}^{N_2}\sum\limits_{l=1}^{N_2}  e^{-\mu'(d_{x_ky_j}+d_{x_ky_l})}\left(\frac{\|\tilde\chi_{j}g\|^2+\|\tilde\chi_{l}g\|^2}{2}\right).\label{Jul.31.eq1}
\end{align}
By the symmetry with respect to $j$ and $l$ in the right-hand side of~\eqref{Jul.31.eq1}, \eqref{Jul.31.eq1} implies 
\begin{align}
M(g)\leq&\sum\limits_{k=1}^{N_1}\sum\limits_{j=1}^{N_2}\sum\limits_{l=1}^{N_2}  e^{-\mu'(d_{x_ky_j}+d_{x_ky_l})} \|\tilde\chi_{j}g\|^2= \sum\limits_{j=1}^{N_2}\|\tilde \chi_{j}g\|^2C_{XY},\label{last: eq2}
\end{align}
where 
\begin{equation}
C_{XY}:=\sum\limits_{k=1}^{N_1}\sum\limits_{l=1}^{N_2} e^{-\mu'(d_{x_ky_j}+d_{x_ky_l})} .\label{def: CXY}
\end{equation}

First, we sum over $l$. To this end, for each $k$, we decompose $\mathbb R^d$ into the spherical shells
\begin{equation}\label{shell}
    \Lambda_{m}(x_k)=\{ x\,:\, r_m \leq d_{xx_k} <r_{m+1} \},
\end{equation}
where $m=0,\cdots,$ $r_m=r_0+m\epsilon d_{XY}, $ with $r_0=(1-\epsilon)d_{XY}$. We sum first over the balls covering a given shell and then over the shells. Each shell $\Lambda_m(x_k)$ is covered by at most $\tilde N_m$ balls $B_l^Y$, with $\tilde N_m= C_n r_m^{n-1}$ for some constant $C_n>0$. This gives
\begin{align}
     C_{X,Y}\leq & \sum\limits_{k=1}^{N_1}\sum\limits_{m=0}^\infty C_n r_{m}^{n-1}e^{-\mu'(d_{x_ky_j}+r_{m})}.
\end{align}

To evaluate the sum over the shells, we use that
\begin{align}
    \sum\limits_{m=0}^\infty r_m^{n-1} e^{-\mu' r_m}\leq & C\sum\limits_{m=0}^\infty e^{-(1-\epsilon/2)\mu'r_m}\nonumber\\
    \leq& Ce^{-(1-\epsilon/2)\mu'r_0}=Ce^{-(1-\epsilon/2)(1-\epsilon)\mu'd_{XY}}.
\end{align}
We conclude that there is a $C_{n,\epsilon,\mu'}>0$ depending on $n$, $\epsilon$ and $\mu',$ s.t.
\begin{align}
  C_{X,Y}\leq C_{n,\epsilon,\mu'} \sum\limits_{k=1}^{N_1} d_{XY}^{n-1}e^{-\mu'(d_{x_ky_j}+d_{XY})}.\label{eq: CXY1}
\end{align}
Next, to estimate the sum over $x_k$, we introduce the spherical shells centered at $y_j$ 
\begin{equation}
    \Lambda_{m}(y_j)=\{ x\,:\, r_m \leq d_{xy_j} <r_{m+1} \},
\end{equation}
where $r_m=r_0+m\epsilon d_{XY}, m=0,\cdots,$ with $r_0=(1-\epsilon)d_{XY}$. Similarly, following~\eqref{eq: CXY1}, we obtain 
\begin{equation}\label{shell2}
    C_{X,Y}\leq C_{n,\epsilon,\mu'}^2 d_{XY}^{2(n-1)} e^{-2\mu'd_{XY}}.
\end{equation}
\par
This, together with Eqs.~\eqref{last: eq2} and~\eqref{def: CXY}, implies 
\begin{equation}
    M(g)\leq  C_{n,\epsilon,\mu'}^2 d_{XY}^{2(n-1)} e^{-2\mu'd_{XY}}\|g\|^2.\label{est: Mg}
\end{equation}
Therefore, using~\eqref{last: eq}, and~\eqref{est: Mg}, we conclude that with $\mu''=(1-2\epsilon)\mu$,
\begin{equation}
    \| \chi_Xe^{-itH}\chi_Yg\| \leq C_{n,\epsilon,\mu'} C_{n,\epsilon,\mu}'e^{\mu ct}e^{-\mu''d_{XY}}\|g\|,\label{cpt2}
\end{equation}
where $C_{n,\epsilon,\mu}':=\sup\limits_{u\geq 0} u^{(n-1)} e^{-\epsilon \mu u}.$ Estimate~\eqref{cpt2} yields 
\begin{equation}
    \|\chi_X U_t \chi_Y\|\leq C  e^{-\mu'' (d_{XY}-c't)},\qquad \mu''=(1-2\epsilon)\mu,\,\, c'=\frac{c}{1-2\epsilon}\label{cpt23}
\end{equation}
with the constant $C$ depending on $\epsilon=1-\frac{\mu''}{\mu}$, $n$ and $\mu$.
This{, with $\mu''$ replaced by $\mu'$,} implies~\eqref{MVB}. 
\end{proof}

\section{
Localization of observables and proof of Theorem~\ref{thm:LRB}}\label{sec:localization}


 We introduce a mathematically convenient notion of localized observables. We say that an observable $A$ is \emph{localized} in $X$ if 
\begin{align}\label{obs-loc}
	A \chi_{X^c}= \chi_{X^c} A=0,\, \text{ or }\,\ A =\chi_{X}A \chi_{X}.
\end{align}
\par Since $\chi_{X^c}=\mathbbm{1}-\chi_X$, Eq.~\eqref{A=Ax+Axc} implies, for any operator $A$ acting on $X$,
\begin{equation}\label{A=tA+1}
    A=\tilde A_X+\mathbbm{1},\qquad \text{where $\tilde A_X=\chi_XA\chi_X-\chi_{X}$. }
\end{equation}
By the definition, the observable $\tilde A_X$ is localized in $X$.

\begin{proof}[Proof of Theorem~\ref{thm:LRB}]\label{sec:LRB}

Since $A$ and $B$ act on $X$ and $Y$, respectively, by~\eqref{A=tA+1}, they are of the form $A=\tilde A_X+\mathbbm{1}$ and $B=\tilde B_Y+\mathbbm 1$, where $\tilde A_X$ and $\tilde B_Y$ are localized in $X$ and $Y$, respectively. Let $A_t=\alpha_t(A)$ and $\tilde A_{t,X}=\alpha_t(\tilde A_X)$. Then $A_t=\tilde A_{t,X}+\mathbbm 1$ and 
\begin{align}
    [A_t, B]=& [\tilde A_{t,X}, \tilde B_Y].\label{LRB: decom}
\end{align}
Using that $\tilde A_X=\chi_X \tilde A_X \chi_X$ and $\tilde B_Y=\chi_Y \tilde B_Y \chi_Y$ and using Theorem~\ref{thm:MVB}, we obtain 
\begin{align}
     \|[\tilde A_{t,X}, \tilde B_Y] \|\leq &\|\tilde A_X \chi_X e^{-itH}\chi_Y \tilde B_Y\|+\|\tilde B_Y \chi_Y e^{itH}\chi_X \tilde A_X\|\nonumber\\
     \leq& Ce^{-\mu (d_{XY}-ct)}\|\tilde A_X\|\|\tilde B_X\|.\label{eq: est: AB1}
\end{align}
Since $X^c\neq \emptyset$ and since $\|Au\|=\|u\|$ for $u$ supported in $X^c$, we have $\|A\|\geq 1$. Hence, 
\begin{equation}
    \|\tilde A_X\|\leq \|\chi_X A\chi_X\|+\|\chi_X\|\leq \|A\|+1\leq 2\|A\|
\end{equation}
and similarly for $\tilde B$. These inequalities together with~\eqref{eq: est: AB1} yield
\begin{equation}
     \|[\tilde A_{t,X}, \tilde B_Y] \|\leq Ce^{-\mu (d_{XY}-ct)}\|A\|\|B\|.\label{eq: 2.7}
\end{equation}

Relations~\eqref{LRB: decom} and~\eqref{eq: 2.7} yield~\eqref{LRB'}.
\end{proof}








\section{MVB for evolution of observables: Proof of Theorem~\ref{thm: LCA}}\label{sec: MVB}

\begin{proof}[Proof of Theorem~\ref{thm: LCA}]
Since $A$ acts on $X,$ by~\eqref{A=tA+1}, it can be be written as $A=\tilde A_X+\mathbbm1$, where the operator $\tilde A_X$ is localized in $X$, i.e. satisfies $\tilde A_X=\chi_X \tilde A_X\chi_X$. Hence the operator family $A_{t,X_\eta}$ defined in~\eqref{def: AtY} can be written as 
\begin{equation}
    A_{t,X_\eta}:=\chi_{X_\eta}\alpha_t(\tilde A_X)\chi_{X_\eta}+\mathbbm 1. 
\end{equation}
By the definition,~$A_{t,X_\eta}$ acts on $X_\eta$. To prove~\eqref{aprxAt}, we use that $A_t=\alpha_t(\tilde A_X+\mathbbm1)=\alpha_t(\tilde A_X)+\mathbbm 1$, to write 
\begin{align}
    A_t-A_{t,X_\eta}=&\alpha_t(\tilde A_X)-\chi_{X_\eta}\alpha_t(\tilde A_X)\chi_{X_\eta}\nonumber\\
    =&\chi_{X_\eta}\alpha_t(\tilde A_X)\chi_{X_\eta^c}+\chi_{X_\eta^c}\alpha_t(\tilde A_X).
\end{align}
Using this relation and Theorem~\ref{thm:MVB}, we arrive at~\eqref{aprxAt}.
\end{proof}
\begin{remark} Theorem~\ref{thm: LCA} yields a natural (but slightly longer) proof of Theorem~\ref{thm:LRB}: $A_{t,X_\eta}$ commutes with $B$ as long as $\eta<d_{XY}$ and therefore $[A_t,B_t]=[R_t^A,B]$, where $R_t^A=A_t-A_{t,X_\eta}$, which leads to an estimate of $[A_t,B_t]$ through an estimate of the remainder $R_t^A$. 
    
\end{remark}

\medskip

\appendix

\section{Differentiability and power estimates}\label{sec: DEs}

In this appendix, we consider a self-adjoint operator $H$ on $L^2(\Lambda)$ under a weaker assumption than the analyticity assumption (A) of the Introduction.

With the definition~\eqref{def: Txi} and $H_\xi=T_\xi H T_\xi^{-1}$, $\xi\in \mathbb R^n$, we assume

(Diff) The family $H_\xi, \xi\in \mathbb R^n$, is $m$ times differentiable, with all derivatives yielding bounded operators.

{W}e define the number $\tilde c$:
\begin{align}
   \tilde c:=&\sum\limits_{k=1}^m \frac{1}{k!} \re (i\mu)^{k-1}\sup\limits_{b\in S^{n-1}} \sup(b\cdot \nabla_\xi)^k H_{\xi}{\big|_{\xi=0}}. \label{def: tc}
\end{align}

\begin{thm}\label{thm:MVB-diff} Suppose that Assumption~(Diff) hold for some $m\geq n$ and let $X$ and $Y$ be two bounded, disjoint sets. Then, for every $ \tilde c'>\tilde c$, there exists a constant $C>0$, depending on $\tilde c'-\tilde c$ and $n$, such that 
   \begin{equation} \label{MVB'} 
 \|\chi_{_{X}}\,e^{-iH t} \chi_{_{Y}}\|\leq \, C  tM (d_{XY}-t\tilde c')^{-m-1+n}, 
 \end{equation}
for all $ 1\leq t\le d_{XY}/\tilde c' $, where constant $M$ is given by 
\begin{equation}
  M:=1+ \sup\limits_{b\in S^{n-1}}\| (b\cdot \nabla_\xi)^{m+1} H_\xi\|.\label{def: M}
\end{equation}
\end{thm} 
\begin{remark} For Hamiltonians of form~\eqref{H}, condition (Diff) follows from the condition
\begin{equation}
\label{om-cond}\text{ $|\p^\alpha\om(k)|\ls 1$ for $1\leq |\alpha|\leq m+1$ for some $m\geq 1$.}
\end{equation}
    
\end{remark}
The proof of Theorem~\ref{thm:MVB-diff} is based on the following proposition.

\begin{prop}\label{prop: 3.1} Let $X=B_r(x_0)$ and $Y=B_r(y_0)$ with $r=\epsilon/2, \epsilon\in (0,1)$, and let $\xi^z=zb+\xi$ with $z=\lambda+i\mu\in \mathbb C^+$ and $b=\frac{x_0-y_0}{|x_0-y_0|}\in S^{n-1}$. For $H_\xi$, $m+1$ times boundedly differentiable, instead of~\eqref{Ut: eq T}, we have for all $\mu\in (0,1)$ and $d_{x_0y_0}-\epsilon-t\tilde c\geq 0$,
\begin{equation}
    \chi_XU_t \chi_Y=\chi_X T^{-1}_{\xi^z} \tilde U_{t,\xi^z} T_{\xi^z} \chi_Y +Rem, \label{re: eq}
\end{equation}
where $\tilde U_{t,\xi^z}$ is the almost analytic extension of $U_{t,\xi^z}$ defined as
\begin{equation}
    \tilde U_{t,\xi^z}=e^{-i\tilde H_{\xi^z}t},\label{def: tUtxiz}
\end{equation}
where 
\begin{equation}
    \tilde H_{\xi^z}=\sum\limits_{k=0}^m \frac{1}{k!} (b\cdot \nabla_\xi)^k H_{\xi^\lambda}(i\mu)^k,\label{def: tildH}
\end{equation}
and $Rem$ is a bounded operator satisfying
\begin{equation}
\|Rem \| \lesssim_\mu \frac{tM}{(d_{x_0y_0}-\epsilon-t\tilde c)^{m+1}},\label{est: Remmu}
\end{equation}
with $\tilde c$ and $M$ defined in Eqs.~\eqref{def: tc} and~\eqref{def: M}, respectively.

\end{prop}
\begin{remark} (i) For $\tilde H_{\xi^z}$ given by~\eqref{def: tildH}, we have 
\begin{equation}
\tilde c=\sup\limits_{b\in S^{n-1}} \sup \frac{\im \tilde H_{\xi^z}}{\mu}.
\end{equation}
(ii) For $\mu\in (0,1)$ sufficiently small, and $b\in S^{n-1}$, we have
\begin{align}
     \tilde c= & \sup\limits_{b\in S^{n-1}} \sup \left(b\cdot \nabla_\xi H_{\xi^\stkout{\lambda}}\right)+O(\mu).
\end{align}
(iii) We can also consider the speed $\tilde c(b)=\sup (\im \tilde H_{\xi^z})/\mu$ in a direction $b\in S^{n-1}$. 
\end{remark}

We derive this proposition from the following two lemmas. 
\begin{lemma}\label{lem: 4.1} Let $f(z)$ be a differentiable function in the strip $\mathcal S_a$, for some $a$, which is independent of $\re z$. Then 
\begin{equation}
    f(z)=f(x)-\int_0^y (\bar\p_z f)(x+is)ds,\label{eq: fz}
\end{equation}
where $z=x+y$. 
\end{lemma}
\begin{proof} Let $z=x+iy$. By the fundamental theorem of Calculus, we have 
\begin{equation}
    f(z)=f(x)-i \int_0^y (\p_y f) (x+is)ds.
\end{equation}
Furthermore, since $f(z)$ is independent of $x$, we have $\p_x f(x+iy)=0$. These two relations and the definition $\bar\p_z=\p_x+i\p_y$ imply~\eqref{eq: fz}. \end{proof}

\begin{lemma}\label{Lem: 4.2} For any $f\in C^{m+1}(\mathbb R)$, define an almost analytic extension of $f$ as 
\begin{equation}
    \tilde f(z)= \sum\limits_{k=0}^m
 f^{(k)}(x) \frac{(iy)^k}{k!},\label{fz: expan}
 \end{equation}
 where $f^{(k)}=\frac{d^k}{dx^k}f$ and $z=x+iy$. Then $\tilde f $ satisfies the estimate 
 \begin{equation}
    |\bar \p_z \tilde f(z)|\leq \frac{1}{m!} |f^{(m+1)}(x)| |y|^m.\label{pbf}
\end{equation}
\end{lemma}

\begin{proof} Eq.~\eqref{fz: expan} follows from the straightforward computation. With $\bar \p_z =\p_x+i\p_y$, $\bar \p_z\tilde f(z)$ reads 
\begin{align}
    \bar \p_z \tilde f(z)=&\sum\limits_{k=0}^m f^{(k+1)}(x)\frac{(iy)^k}{k!}+ i\sum\limits_{k=1}^m f^{(k)}(x) \frac{i(iy)^{k-1}}{(k-1)!}\nonumber\\
    =& \frac{ 1}{m!} f^{(m+1)}(x)(iy)^m,\label{eq: 4.7}
\end{align}
which gives~\eqref{pbf}.\end{proof}
\begin{proof}[Proof of Proposition~\ref{prop: 3.1}] Recall that $X=B_r(x_0-y_0)$ and $Y=B_r(0)$, where $r=\epsilon/2$. Now, let $g(\lambda)=\chi_{X}T_{\xi^\lambda}^{-1} f(\lambda) T_{\xi^\lambda} \chi_Y$, where $f(\lambda):=U_{t,\xi^\lambda}$. We define 
\begin{equation}
    \tilde g(z)=\chi_X T_{\xi^{z}}^{-1} \tilde f(z) T_{\xi^{z}}\chi_Y,\label{def: tg}
\end{equation}
where $\tilde f(z)$ is the almost analytic extension of $f(\lambda)$ in $\lambda$ constructed in Eq.~\eqref{def: tUtxiz}.

To compute $\bar \p_z g(z)$, we note that, by Lemma~\ref{Lem: 4.2},
\begin{equation}
    \bar\p_z \tilde H_{\xi^z}= \frac{1}{m!} \p_\lambda^{m+1} H_{\xi^\lambda}(i\mu)^m.
\end{equation}
This and the Duhamel principle yield
\begin{align}
    \bar\p_z f(z)=&  \frac{-i}{m!}\int_0^t  e^{-i\tilde H_{\xi^z}(t-s)} \bar \p_z \tilde H_{\xi^z} e^{-i\tilde H_{\xi^z}s} ds\nonumber\\
    =& -i\int_0^t e^{-i\tilde H_{\xi^z}(t-s)}  \p_\lambda^{m+1} H_{\xi^\lambda} e^{-i\tilde H_{\xi^z}s}ds(i\mu)^m,\label{eq: fpz}
\end{align}

Since $\chi_X T_{\xi^{z}}^{-1}$ and $T_{\xi^{z}}\chi_Y$ are analytic in $z$, we have, by the Leibnitz rule and Eq.~\eqref{eq: fpz}, that 
\begin{align}
    \bar \p_{z} \tilde g(z)=&\chi_X T_{\xi^{z}}^{-1}\bar \partial_{z} f(z) T_{\xi^{z}}\chi_Y\nonumber\\
    =&\frac{ -i}{m!}\chi_X T_{\xi^{z}}^{-1} \int_0^t R_{t,s}(z)(i\mu)^mds,\label{eq: 3.9}
\end{align}
where, with $z=\lambda+i\mu$, 
\begin{equation}
    R_{t,u}(z):=\chi_X T_{\xi^{z}}^{-1}  e^{-i\tilde H_{\xi^z}(t-u)}  \p_\lambda^{m+1} H_{\xi^\lambda} e^{-i\tilde H_{\xi^z}u} T_{\xi^{z}}\chi_Y.\label{Rtu: def}
\end{equation}
\par Next, we claim that $\tilde g(z)$ is independent of $\re z$. Indeed, we have 
\begin{equation}
    \tilde g(z)=\chi_X T_{\xi^{z}}^{-1} T_{\eta}^{-1}T_{\eta} \tilde f(z) T_{\eta}^{-1}T_{\eta} T_{\xi^{z}}\chi_Y,\label{tgzeta1: eq}
\end{equation}
where $\eta=\alpha b\in \mathbb R^n$ with $\alpha\in \mathbb R$. Using that $T_{\eta}$ commutes with $\p_\xi$, we find 
\begin{equation}
    T_{\eta}\p_\lambda^k H_{\xi^\lambda}T_{\eta}^{-1}=\p_{\lambda}^k[ T_{\eta}H_{\xi^\lambda}T_{\eta}^{-1}]=\p_{\lambda}^kH_{\xi^{\lambda+a}},\qquad k=0,\cdots,m.
\end{equation}
This, together with Eqs.~\eqref{def: tildH},~\eqref{def: tUtxiz} and $\tilde f(z)=\tilde U_{t,\xi^z}$, yields that $T_{\eta}\tilde f(z)T_{\eta}^{-1}=\tilde f(z+\alpha)$. The last relation, together with~\eqref{tgzeta1: eq} and the group property $T_{\eta}T_{\xi^{z}}=T_{\xi^{z+\alpha}}$, implies 
\begin{equation}
    \tilde g(z)=\tilde g(z+\alpha), \qquad \forall\, \alpha\in \mathbb R,
\end{equation}
which shows that $\tilde g(z)$ is independent of $\re z$. Hence, Lemma~\ref{lem: 4.1} applies to $\tilde g(z)$ and yields
\begin{equation}
    g(\lambda)=\tilde g(\lambda)=\tilde g(z)+Rem,\label{g: tg+R}
\end{equation}
where, by~\eqref{eq: 3.9},  
\begin{align}
    Rem:=& i\int_0^\mu   (\bar\p_z \tilde g)(x+is)ds\nonumber\\
    =&\frac{1}{m!}\int_0^\mu\int_0^t R_{t,u}(\lambda+is)(is)^mduds.\label{rem}
\end{align}
Our next goal is to estimate this reminder.\par To estimate $ R_{t,u}(\lambda+is)$, we proceed as in Proposition~\ref{prop: E} to obtain
\begin{equation}
    \| e^{-i\tilde H_{\xi^z}t}\|\leq e^{\mu t \tilde c},\label{fm+10}
\end{equation}
where $\mu=\im z$ and the constant $\tilde c$ is defined in Eq.~\eqref{def: tc}. Next, from Eq.~\eqref{Rtu: def} we find, 
\begin{align}
    \|R_{t,u}(\lambda+is)\|\leq &\|\chi_XT_{\xi^{\lambda+is}}^{-1}\|\| e^{-i\tilde H_{\xi^{\lambda+is}}(t-u)}\| \| \p_\lambda^{m+1} H_{\xi^\lambda}\|\| e^{-i\tilde H_{\xi^z}u}\|\| T_{\xi^{\lambda+is}}\chi_Y\|.\label{eq: remark1}
\end{align}
Proceeding as~\eqref{YUtX: est}-\eqref{1.32}, using~\eqref{fm+10} and~\eqref{def: M} and assuming $d_{XY}-\epsilon-t\tilde c\geq 0$, we obtain 
\begin{equation}
   \|R_{t,u}(\lambda+is)\| \leq  e^{-s(d_{x_0y_0}-\epsilon-t\tilde c)}  M.\label{B.26}
\end{equation}
This together with Eqs.~\eqref{Rtu: def} and~\eqref{rem} and estimates~\eqref{B.26} and, for all $|x_0-y_0|-\epsilon-t\tilde c>0$ and $\mu\in (0,1)$, 
\begin{equation}
  \frac{1}{m!}  \int_0^\mu e^{-s(d_{x_0y_0}-\epsilon-t\tilde c)} s^mds\lesssim_\mu \frac{1}{(d_{x_0y_0}-\epsilon-t\tilde c)^{m+1}},
\end{equation}
yields
\begin{equation}
    \|Rem\|\lesssim_\mu \frac{tM}{(d_{x_0y_0}-\epsilon-t\tilde c)^{m+1}},\label{Rem: est}
\end{equation}
where $M$ is given by~\eqref{def: M}, and subsequently
\begin{equation}
    g(\lambda)=\tilde g(z)+Rem,
\end{equation}
with $Rem$ satisfying~\eqref{Rem: est}, yielding~\eqref{re: eq}-\eqref{est: Remmu}. \end{proof}

\begin{proof}[Proof of Theorem~\ref{thm:MVB-diff}] Let $X, Y$ and $b$ be the same as in Proposition~\ref{prop: 3.1}. We estimate the first term on the right-hand side of~\eqref{re: eq} as in the proof of Theorem~\ref{thm:MVB}. Similarly to \eqref{1.32}, we find
\begin{equation}
    \| \chi_X T_{\xi^z}^{-1} \tilde U_{t,\xi^z} T_{\xi^z} \chi_Y\|\leq \| \tilde U_{t,\xi^z}\| e^{-\mu'd_{x_0y_0}},\label{July23: eq1}
\end{equation}
with $\xi^z=zb+\xi^\stkout{\perp}$ \stkout{(see~\eqref{decom: xiz})}, $b=\frac{x_0-y_0}{|x_0-y_0|}$, $\tilde U_{t,\xi^z}$ defined in Eq.~\eqref{def: tUtxiz} and $\mu'=(1-\epsilon)\mu$. Next, applying~\eqref{fm+10} and setting $\tilde c'=\tilde c/(1-\epsilon)$, we conclude that
\begin{equation}
 \|\chi_X T_{\xi^z}^{-1} \tilde U_{t,\xi^z} T_{\xi^z} \chi_Y\|   \leq e^{-\mu'(d_{x_0y_0}-\tilde c't)}.\label{July23: eq4}
\end{equation}
Eq.~\eqref{re: eq}, together with estimate~\eqref{July23: eq4}, yields
\begin{equation}
    \| \chi_X U_t \chi_Y\|\lesssim_\mu e^{-(1-\epsilon)\mu (d_{x_0y_0}-\tilde c t)}+\frac{\langle t\rangle  M}{(d_{x_0y_0}-\epsilon-t\tilde c)^{m+1}}.\label{Qeq}
\end{equation}

We take $\mu=\frac{1}{2}$. Then, for $1< t< \frac{d_{x_0y_0}}{\tilde c'}, \, \tilde c'=\tilde c+2\epsilon>\tilde c$, we have 
\begin{equation}
    \| \chi_{X} e^{-itH}\chi_{Y}\|\lesssim  \frac{t M}{(d_{x_0y_0}-t\tilde c')^{m+1}}\label{A: MVB}.
\end{equation}
\par Now, we return to general compact sets $X$ and $Y$. Covering $X$ and $Y$ with balls $B_r(x_j)$, $j=1,\cdots,N_1$, and $B_r(y_k), k=1,\cdots,N_2$ and proceeding as in \eqref{cpt3}-\eqref{def: CXY} and using~\eqref{A: MVB}, we obtain
\begin{equation}
    \|\chi_X e^{-itH}\chi_Yg\|^2\leq C\sum\limits_{j=1}^{N_2}t^2M^2\|\tilde\chi_jg\|^2\tilde C_{XY},\qquad \forall\, g\in L^2(\mathbb R^n),
\end{equation}
where 
\begin{equation}
    \tilde C_{XY}:=\sum\limits_{k=1}^{N_1}\sum\limits_{l=1}^{N_2}  (d_{x_ky_j}-t\tilde c')^{-m-1}(d_{x_ky_l}-t\tilde c')^{-m-1}.
\end{equation}
Using the shell argument, as in~\eqref{shell}-\eqref{shell2}, we arrive at
\begin{equation}
     \|\chi_X e^{-itH}\chi_Yg\|^2\lesssim_n t^2M^2(d_{XY}-t\tilde c')^{-2m-2+2n}\|g\|^2,
\end{equation}
which yields~\eqref{MVB'}.\end{proof}

\section{{$N$}-particle dynamics}\label{sec: MVB-N}

\noindent For the $N$-particle problem, consider the quantum Hamiltonian for $N$ identical bosons 
\begin{equation}
    H_N:=\sum\limits_{j=1}^N (\omega_1(p_j)+v(x_j))+\frac{1}{2}\sum\limits_{i\neq j} w(x_i-x_j)\label{def: HN}
\end{equation}
on $L^2_{sym}(\mathbb R^{dN})$, where $L^2_{sym}(\mathbb R^{dN})$ is either bosonic space of symmetric functions of ferminoic one associated with certain representations of the symmetric group $S_N$ of permutations of $N$ indices (see e.g. \cite{GS}).

The operator $H_N$ is of the form~\eqref{H}, with
\begin{equation}
\omega(k)=\sum\limits_{j=1}^{N}\omega_1(k_j)\,\text{ and 
}\, V(x)=\sum\limits_{j=1}^N v(x_j)+\frac{1}{2}\sum\limits_{i\neq j} w(x_i-x_j),\label{def: N: omegaV}
\end{equation}
and $x=(x_1,\cdots,x_N)\in \mathbb R^{dN}$, $k=(k_1,\cdots,k_N)\in \mathbb R^{dN}$. Now, we define $\chi_X(x)$ as the characteristic function of the set $X^N:=\underbrace{X\times \cdots \times X}_{N\text{-fold product}}$:
\begin{equation}
\tilde\chi_X(x)\equiv \chi_{X^N}(x):=\prod\limits_{j=1}^{N}\chi_X(x_j).\label{Nfold}
\end{equation}
\begin{thm}\label{thm:MVB-N}Let $H_N$ be as in Eq.~\eqref{def: HN} and let $\omega_1$ and $V$ ( defined in Eq.~\eqref{def: N: omegaV}) satisfy Condition~\eqref{om-cond'} (with $n=d$). Let $X$ and $Y$ be two disjoint sets. Then, for any $\mu'\in (0,\mu)$, we have
\begin{equation}
    \| \tilde\chi_X e^{-itH_N}\tilde\chi_Y\|\leq C  e^{-\mu'N(d_{XY}-c_1't)},\label{goal: N}
\end{equation}
where $X$ is the multiplication operator by the function $\chi_X(x)$ in~\eqref{Nfold}, $C>0$ is a constant depending on $\epsilon=1-\frac{\mu'}{\mu}, d, N$ and $\mu$, with $c'_1=\frac{c_1}{1-\epsilon}$ and
\begin{equation}
    c_1:=\sup\limits_{\xi\in \mathbb R^d,\, b\in S^{d-1}} \frac{\im \omega_1(\xi+i\mu b)}{\mu}<\infty.
\end{equation}
\end{thm}

Eq.~\eqref{goal: N} is one of the simplest many-body estimates. The more refined and more difficult estimate to prove would be one with $\chi_X(x)$ and $\chi_Y(y)$ replaced in~\eqref{goal: N} by $\chi_X^{(k)}(x)\equiv \chi_{X^k}(x)=\prod\limits_{j=1}^k \chi_X(x_j)$ and $\chi_Y^{(l)}(y)\equiv \chi_{Y^l}(y)=\prod\limits_{j=1}^l \chi_Y(y_j),$ with $1\leq k\leq l<N.$
\begin{proof}[Proof of Theorem~\ref{thm:MVB-N}] We follow the proof of Theorem~\ref{thm:MVB}, with the following modifications:
\begin{equation}
\xi=(\xi_1,\cdots,\xi_N)\in \mathbb R^{dN},\, \xi_j\in \mathbb R^d,\, b=(b_1,\cdots,b_N), b_j\in S^{n-1}, \,\forall\, j, 
\end{equation}
\begin{equation}
\zeta=(\zeta_1,\cdots,\zeta_N)\in \mathbb R^{dN},\, \zeta_j=\xi_j+i\mu b_j,\,\mu\in (0,a),\, b\cdot x=\sum\limits_{j=1}^N b_j\cdot x_j.
\end{equation}
As in the proof of Theorem~\ref{thm:MVB}, we reduce proving Eq.~\eqref{goal: N} for disjoint sets $X$ and $Y$ to Eq.~\eqref{goal: N} for the polyballs $X=\prod\limits_{j=1}^NX_j,\, X_j:=B_r(x_{0j}-y_{0j})$ and $Y=\prod\limits_{j=1}^N Y_0, \, Y_0:=B_r(0)$, with $r=\frac{\epsilon}{2}d_{XY}$, with $x_{0j}\in X$ and $y_{0j}\in Y,\, \forall \, j$, and $\epsilon\in (0,1)$.\par Similarly to~\eqref{eq: YUX} and~\eqref{est: Utxi: zeta}, we have 
\begin{equation}
\|\tilde \chi_X e^{-iH_Nt}\tilde\chi_Y\|\leq e^{-\mu (r_X-r_Y)}e^{\mu tc},\label{thm1.2: eq20}
\end{equation}
with 
\begin{align}
c=&\frac{1}{\mu}\sup\limits_{b\in \bigotimes\limits_{1}^NS^{d-1} }\sup\limits_{\xi\in \mathbb R^{dN}}\im \sum\limits_{j=1}^N\omega(\xi_j+i\mu b_j),
\end{align}
\begin{equation}\label{def: rYrX0}
    r_Y=\sup\limits_{y\in Y}b\cdot y , \quad r_X=\inf\limits_{x\in X} b\cdot x.
\end{equation}
Let $b=(b_1,\cdots,b_N)$, with $b_j=\frac{x_{0j}-y_{0j}}{|x_{0j}-y_{0j}|}, j=1,\cdots,N$. We claim that
    \begin{equation}
        r_X-r_Y\geq  (1-\epsilon)\sum\limits_{j=1}^N|x_{0j}-y_{0j}|\geq (1-\epsilon)Nd_{XY}. \label{thm1.2: eq10}
    \end{equation}
Indeed, by the definitions of $X$ and $Y$ and Eq.~\eqref{def: rYrX0}, we have
\begin{equation}
    r_Y\leq \sum\limits_{j=1}^N |y_j|\leq \frac{\epsilon N}{2}d_{XY}
\end{equation}
and 
\begin{align}
    r_X=& \sum\limits_{j=1}^N b\cdot (x_{0j}-y_{0j}) +\inf\limits_{x\in X} b\cdot (x-(x_{0j}-y_{0j}))\nonumber\\
    \geq & \sum\limits_{j=1}^N (|x_{0j}-y_{0j}|-\frac{\epsilon}{2}d_{XY}).
\end{align}
Hence,
\begin{align}
    r_X-r_Y\geq \sum\limits_{j=1}^N (|x_{0j}-y_{0j}|-\epsilon d_{XY})\geq (1-\epsilon)\sum\limits_{j=1}^N |x_{0j}-y_{0j}|,
\end{align}
as claimed.

Furthermore, we compute
\begin{align}
c=&\frac{1}{\mu}\sum\limits_{j=1}^N \sup\limits_{\xi_j\in \mathbb R^d,\, b_j\in S^{d-1}} \im \omega_1(\xi_j+i\mu b_j)= Nc_1,\label{def: N: c0}
\end{align}
where, recall, 
\begin{equation}
c_1:=\sup\limits_{\lambda \in \mathbb R^n,\, b\in S^{n-1}}\frac{\omega_1(\lambda +i\mu b)}{\mu}.
\end{equation}\par

Combining~\eqref{thm1.2: eq20},~\eqref{thm1.2: eq10} and~\eqref{def: N: c0}, we arrive at 
\begin{equation}
    \| \tilde\chi_X e^{-itH_N}\tilde\chi_Y\|\leq e^{-\mu (1-\epsilon)N(d_{XY}-c'_1t)},
\end{equation}
where $c'_1=\frac{c_1}{1-\epsilon}$, which implies, with $\mu''=(1-\epsilon)\mu$, 
\begin{equation}
     \| \tilde\chi_X e^{-itH_N}\tilde\chi_Y\|\leq e^{-\mu''N(d_{XY}-c'_1t)},\label{rev: eq10}
\end{equation}
with $X=\prod\limits_{j=1}^NX_j,\, X_j:=B_r(x_{0j})$ and $Y=\prod\limits_{j=1}^N Y_0, \, Y_0:=B_r(y_{0j})$, where $r=\frac{\epsilon}{2}d_{XY}$, for some $x_0\neq y_0$, with $x_{0j}\in X$ and $y_{0j}\in Y,\, \forall \, j$, and $\epsilon\in (0,1)$. Proceeding as in~\eqref{cpt3}-\eqref{cpt23} in each variable $x_j\in \mathbb R^d$, $j=1,\cdots,N$, we arrive at~\eqref{MVB} by taking $\mu'=\mu(1-2\epsilon)$ for any $\epsilon\in (0,1)$, and the constant $C$ in~\eqref{goal: N} depends on $\epsilon=1-\frac{\mu'}{\mu}, d, N$ and $\mu$.\end{proof}

\paragraph{\bf Acknowledgment}

The first author is grateful to Jeremy Faupin, Marius Lemm and Jingxuan Zhang, and both authors, to Avy Soffer, for enjoyable and fruitful collaborations. The authors are grateful to the anonymous referees for useful suggestions, for pointing out reference \cite{BD}. 
  The research of I.M.S. is supported in part by NSERC Grant NA7901.
X. W. is supported by 2022 \emph{Australian Laureate Fellowships} grant FL220100072. Her research was also supported in part by NSERC Grant NA7901.  
Parts of this work were done while the second author was at the Fields Institute  for Research in Mathematical Sciences, Toronto.

\medskip

\paragraph{\bf Declarations}
\begin{itemize}
	\item Conflict of interest: The Authors have no conflicts of interest to declare that are relevant to the content of this article.
	\item Data availability: Data sharing is not applicable to this article as no datasets were generated or analysed during the current study.
\end{itemize}

\bigskip

{
\paragraph{\bf Note added in proofs}

M. Lemm observed to us that our approach in the proof of Theorem 1.1 is related to the Combes-Thomas argument. (See e.g. M. Aizenman and S. Warzel, Random Operators. AMS Press 2015) We are grateful to to M. Lemm for this communication and illuminating discussions.
}

	\end{document}